\documentclass[pdftex]{article} 
\def\Title {Optimal acceptance sampling for modules F and F1 of the European Measuring Instruments Directive}

\usepackage[a4paper]{geometry}
\usepackage{amsmath,amsfonts,amssymb,bm,amsthm}
\usepackage[font=small,labelfont=bf]{caption}
\newtheorem{prop}{Proposition}

\usepackage{graphicx,multicol}
\usepackage{array,booktabs}
\newcolumntype{C}[1]{>{\centering\let\newline\\\arraybackslash\hspace{0pt}}p{#1}}
\newcolumntype{R}[1]{>{\raggedleft\let\newline\\\arraybackslash\hspace{0pt}}p{#1}}

\usepackage[table]{xcolor}

\usepackage[pdftex]{hyperref}
\definecolor{darkblue}{rgb}{0,0,.75}
\hypersetup{colorlinks=true, breaklinks=true,
linkcolor=darkblue, menucolor=darkblue, 
urlcolor=darkblue,citecolor=darkblue}

\newcommand{\be}{\begin{equation}}
\newcommand{\ee}{\end{equation}}
\newcommand{\Pac}{P_\text{ac}} 
\newcommand{\pa}{p_\text{a}}
\newcommand{\pb}{p_\text{b}}
\newcommand{\qa}{q_\text{a}}
\newcommand{\qb}{q_\text{b}}
\newcommand{\Pa}{P_\text{a}}
\newcommand{\Pb}{P_\text{b}}
\newcommand{\MIDa}{MID$_\text{a}$}
\newcommand{\MIDb}{MID$_\text{b}$}

\newcommand{\rev}[1]{#1}

\begin{document}

\title{\Title}
\author{%
Cord A.\ M\"uller \thanks{cord.mueller@lmg.bayern.de}\\
German Academy of Metrology (DAM),\\ Bavarian State Office for Weights and Measures (LMG),\\
Wittelsbacherstr.~17, 83435 Bad Reichenhall, Germany
}
\date{\today}

\maketitle

\begin{abstract} 
Acceptance sampling plans offered by ISO 2859-1 are far from optimal under the conditions for statistical verification in modules F and F1 as prescribed by Annex II of the Measuring Instruments Directive (MID) 2014/32/EU, resulting in sample sizes that are larger than necessary.
An optimised single-sampling scheme is derived, both for large lots using the binomial distribution and for finite-sized lots using the exact hypergeometric distribution, resulting in smaller sample sizes that are economically more efficient while offering the full statistical protection required by the MID. 
\end{abstract} 

\tableofcontents 

\newpage

\begin{multicols}{2}

\section{Introduction}
\label{sec:intro} 

\begin{figure*}
\centering
\includegraphics[width=0.9\textwidth]{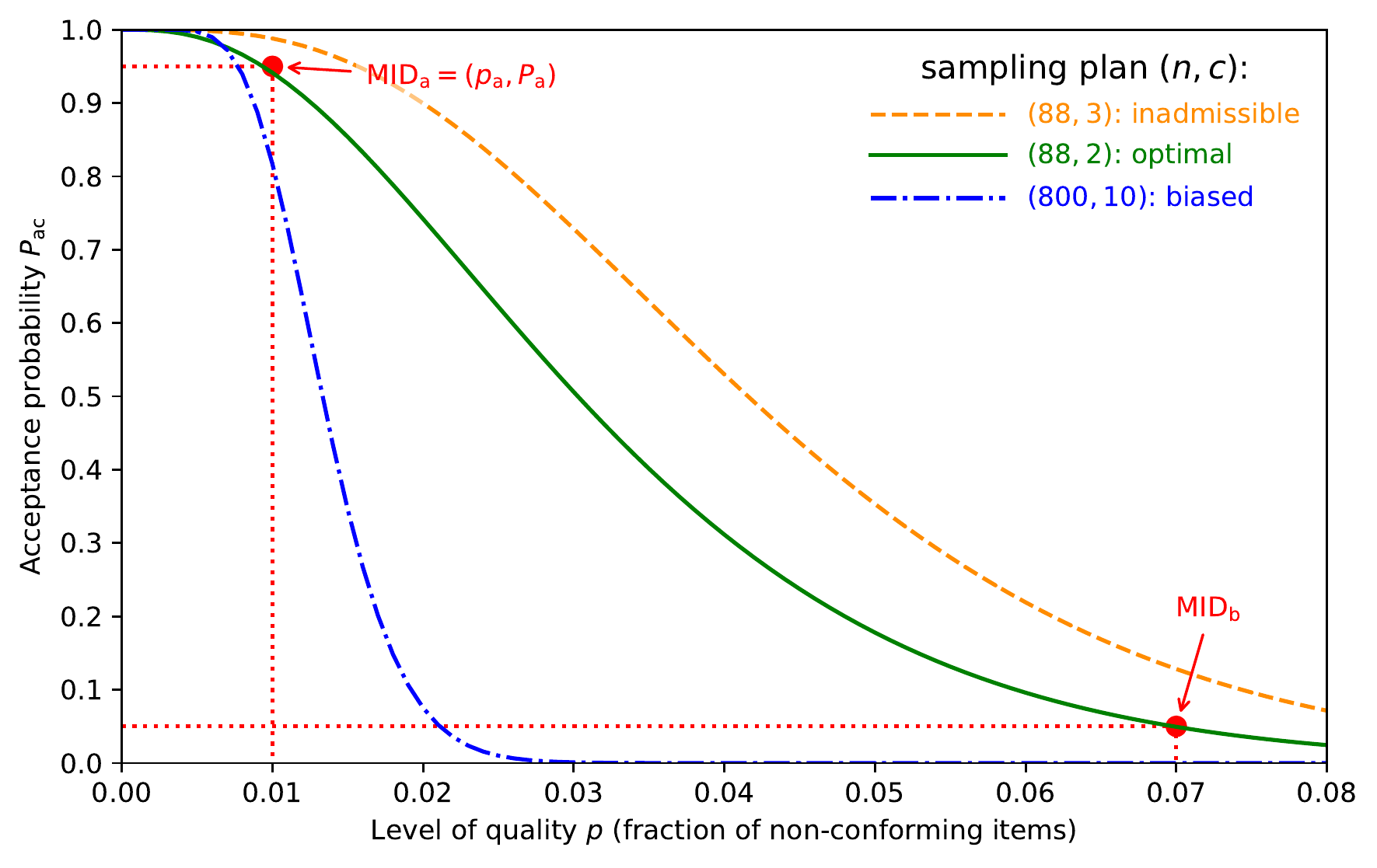}
\caption{%
Acceptance probability $\Pac$ as function of quality level $p$---the operating characteristic (OC)---of different single-sampling plans.   
An OC that does not lie below both points $\text{\MIDa{}}=(\pa,\Pa)=(0.01,0.95)$ and $\text{\MIDb{}}=(\pb,\Pb)=(0.07,0.05)$ is inadmissible according to the conditions \eqref{MID3} and \eqref{MID4}. 
Optimal (unbiased) OCs are close to, and just below of, the two MID points. 
The three OCs shown are plotted using the binomial probability distribution [eq.~\eqref{binomialcgeneral}] for single sampling of $n$ items with  
acceptance number $c$. 
The plan $(n=88,c=3)$ is inadmissible, $(88,2)$ is optimal, and $(800,10)$
is biased toward enforcing a much better quality level than required by the MID.  
} 
\label{fig_OC_MID_stats}  
\end{figure*} 

The European Measuring Instruments Directive 2014/32/EU (MID) \cite{MID} states in Annex II that the statistical verification of products in modules F and F1 must respect the following conditions:

    ``The statistical control will be based on attributes. The sampling system shall ensure: 
\begin{itemize}
      \item[(a)] a level of quality corresponding to a probability of acceptance of 95\%, with a non-conformity of less than 1\%;
      \item[(b)] a limit quality corresponding to a probability of acceptance of 5\%, with a non-conformity of less than 7\%.''  
\end{itemize} 

These wordings must be cast into mathematical equations in order to determine acceptance sampling plans to be used in practice\rev{, as described by the pioneering work of Dodge and Romig \cite{Dodge1929}}. 
There appears to be consensus among the legal bodies in charge of administering the  tests \cite{Welmec8.10} that the MID conditions should be interpreted as 
\begin{align}
 \Pac(p) &= 0.95 \ \Rightarrow \  p < 0.01, \label{MID1} \\
 \Pac(p) &= 0.05 \ \Rightarrow \ p < 0.07 \label{MID2}. 
\end{align}  
Here $p$ is the quality level or ``fraction defective",  namely the fraction of non-conforming%
\footnote{We use the terms ``non-conforming'' and ``defective'' interchangeably for items that fail testing.}
items in the lot, and $\Pac$ is the acceptance probability, a property of the statistical sampling plan to be devised. 
The first quality level $\pa=0.01=1\%$ is commonly called the \emph{acceptance quality limit} (AQL). 
The complement of the acceptance probability at this point, $\alpha=1-\Pac(\pa)$, is known as the \emph{producer's risk} that a lot with this acceptable level of quality is rejected. 
The second quality level 
$\pb=0.07=7\%$ is known as the \emph{limiting quality} (LQ), and the probability of acceptance $\Pac(\pb)=\beta$ is the \emph{consumer's risk} of accepting a lot with this doubtful quality.

Since a lower quality level (i.e., larger $p$) should result in a lower acceptance probability, a valid function $\Pac(p)$ is strictly decreasing: $p>q \Leftrightarrow \Pac(p)<\Pac(q)$.%
\footnote{For finite lot sizes, $\Pac$ may not be \emph{strictly} monotonic, but just monotonic. The MID conditions \eqref{MID1} and \eqref{MID2}, however, implicitly assume an infinite lot size; for details see Appendix \ref{app:MID_finiteN} below.} 
Therefore, the two conditions \eqref{MID1} and \eqref{MID2} for a certain sampling plan to be valid can be formulated equivalently 
\begin{align}
\Pac(\pa)  & <  0.95=\Pa,   \label{MID3}\\
\Pac(\pb)  & < 0.05=\Pb.  \label{MID4}
\end{align}  
In other words, the MID conditions in the prevailing interpretation require  
\begin{enumerate} 
\item[(a)] a lot with AQL $\pa=1\%$ to imply a producer's risk 
\be 
    \alpha=1-\Pac(\pa)> 5\%, 
    \label{MIDa}
\ee 
\item [(b)] and a lot with LQ $\pb=7\%$ to imply a consumer's risk
\be
\beta=\Pac(\pb) < 5\%.
\label{MIDb}
\ee
\end{enumerate}

As a consequence, 
the graph of $\Pac(p)$, the so-called operating characteristic (OC), must lie below and left of the two points $\text{\MIDa{}}=(\pa,\Pa)$ and $\text{\MIDb{}}=(\pb,\Pb)$. 
Figure \ref{fig_OC_MID_stats} shows, as an example, 3 OCs of different single-sampling plans for very large lots, based on the binomial model of Section \ref{sec:Ninfty}. 
Too small samples are typically ruled out because their OCs violate the MID conditions by lying above at least one MID point, as illustrated by the plan with sample size $n=88$ and acceptance number $c=3$. 
In contrast, large samples and large acceptance numbers will typically result in OCs that are far below the MID points; thus they are certainly admissible. 
As an example, Fig.~\ref{fig_OC_MID_stats} shows the OC of $n=800$, $c=10$, a member of the ISO 2859-1 family \cite{ISO2859-1} that is recommended officially \cite{Welmec8.10} for lot sizes from 150\,001 till 500\,000.  
However, such a sampling is clearly biased toward a much better quality level than required by the MID conditions.       

Sampling plans favouring very low AQLs are deemed admissible because the prevailing interpretation of the MID conditions implies a producer's risk \emph{larger} than 5\% [Eq.~\eqref{MIDa}]. 
The main point of the present work is that a sampling plan that is fair and economically acceptable for producers should not impose arbitrary conditions on the producers much stricter than required by the MID, while respecting the legitimate interests of the consumers, of course.  
Therefore, MID-optimised sampling plans outside the scope of ISO 2859-1 are derived in the following sections, devoted to very large lots (Sec.~\ref{sec:Ninfty}) and finite-sized lots (Sec.~\ref{finiteN.sec}), respectively. 
Readers not interested in details of the mathematical derivation are invited to skip to Sec.~\ref{sec:summary}, where a simplified single-sampling scheme optimised for MID modules F and F1 is proposed.   

The concluding Sec.~\ref{sec:outlook} finishes on the observation that it would seem more reasonable if the sampling contract between producer and consumer bounded both their risks from above, guaranteeing both $\beta <5\%$ and $\alpha<5\%$.  
This alternative interpretation of the MID's AQL condition has interesting consequences that are briefly outlined, with details relegated to a follow-up paper.

\section{MID-optimised sampling plans for large lots} 
\label{sec:Ninfty}

\begin{figure*}\centering
\includegraphics[width=.95\textwidth]{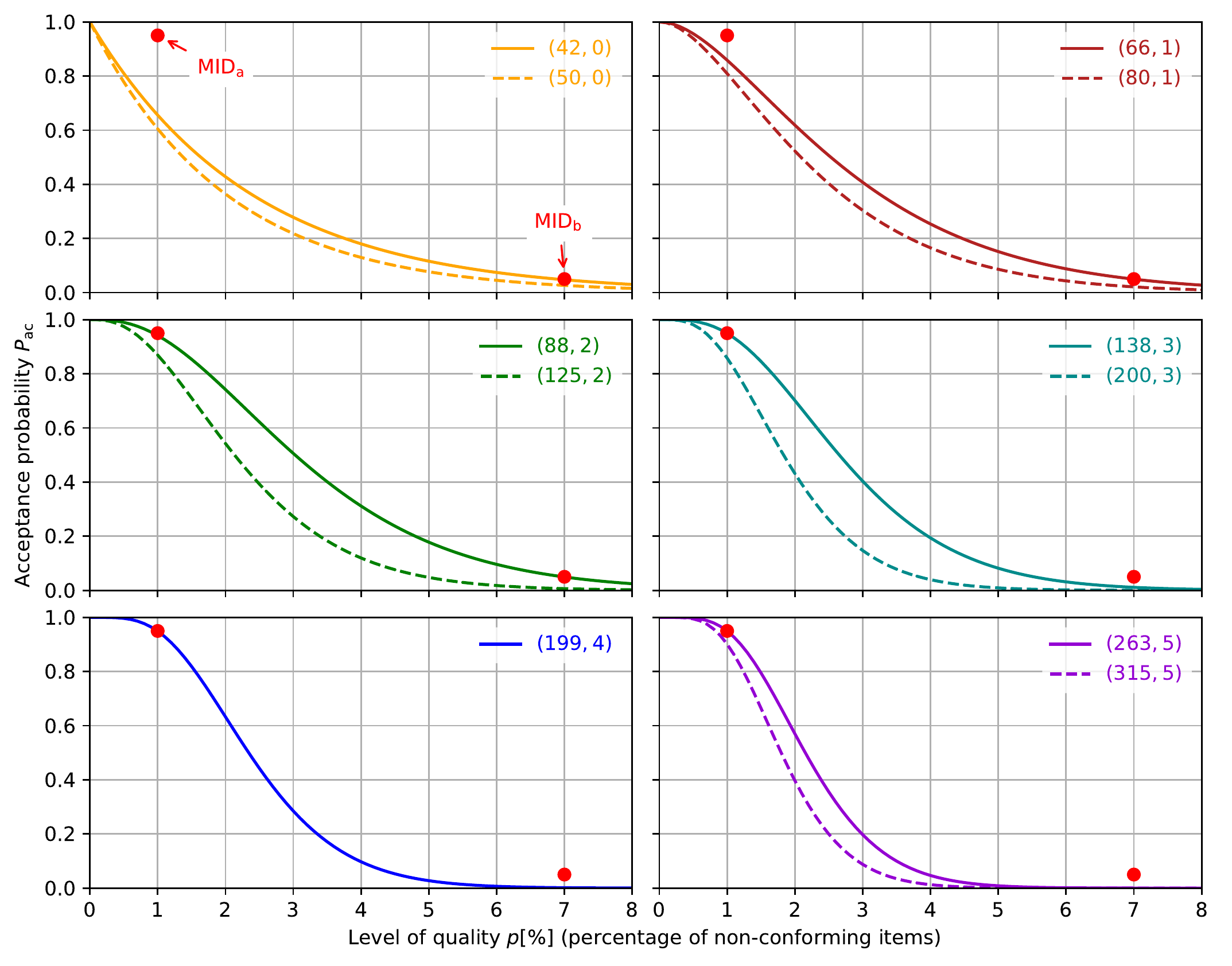}
\caption{Acceptance probability $\Pac(p,n,c)$ of various single-sampling plans $(n,c)$ as function of quality level $p$ based on the binomial distribution Eq.~\eqref{binomialcgeneral}.  
For each acceptance number $c$, the dashed curve shows the OC of the admissible standard ISO 2859-1 member \cite{Welmec8.10}, while the full curve shows the OC of the MID-minimised sample size.
For small acceptance numbers $c<2$, the LQ condition \MIDb{} is the more stringent (OCs are not steep enough). For large acceptance numbers $c>2$, the AQL condition \MIDa{} becomes more stringent (OCs are steeper than required). The optimal sample plan, defined by the smallest sample size $n$ while being unbiased with respect to both MID conditions, is $(n,c)=(88,2)$. 
} 
\label{fig_OC_opt_MID_ISO_matrix}  
\end{figure*}

In a first step, let us consider single sampling where $n$ items are drawn from a very large lot with a constant (but generally unknown) probability $p$ to be defective.%
\footnote{Strictly, a probability $0<p<1$ for a draw without replacement can only stay constant when the lot size $N$ is infinite; for now we assume $N\gg n$ and consider corrections due to finite lot size in Sec.~\ref{finiteN.sec} below.}
The entire lot is accepted if the number of defective items discovered by testing is not larger than the acceptance number $c=0,1,2,\dots$. The acceptance probability to find at most $c$ defective items, each found independently from the others with identical probability $p$,  
in a sample of size $n$ is the cumulative binomial distribution \cite{SchillingNeubauer2017,Mathews2010} 
\be
\label{binomialcgeneral}
\Pac(p,n,c) 
= \sum_{k=0}^{c}\binom{n}{k}p^k(1-p)^{n-k}, 
\ee
where the binomial coefficient $\tbinom{n}{k}=n!/[k!(n-k)!]$ counts the number of different choices of $k$ items among $n$. 

\begin{table*}
\centering
\setlength{\tabcolsep}{0pt}
\begin{tabular}{*{2}{C{2.5em}}|*{5}{C{6.5em}}}
$n$	& 	$c$ & $\Pac(\pa)$  & $\alpha=1-\Pac(\pa)$ & $\qa$ & $\beta=\Pac(\pb) $  & $\qb$ \\
&		& \% & \% & \% & \%  & \% \\
\hline
\rowcolor{gray!25}	42	&	0	& 65.6 	& 34.4	& 0.122 & \textbf{4.75}  & \textbf{6.88} \\
 	                    50	&	0	& 60.5 	& 39.5 	& 0.103 & 2.66  & 5.82 \\
\rowcolor{gray!25}	66	&	1	& 85.9 	& 14.1 	& 0.541 & \textbf{4.96}  & \textbf{6.99} \\
 	                    80	&	1	& 80.9 	& 19.1 	& 0.446 & 2.11  & 5.79 \\
\rowcolor{gray!25}	88	&	2	& \textbf{94.1} 	&  \textbf{5.87} 	& \textbf{0.936} & \textbf{4.94}  & \textbf{6.98} \\
                        125	&	2	& 86.9 	& 13.1 	& 0.657 & 0.62  & 4.95 \\
\rowcolor{gray!25}   138	&	3	& \textbf{94.9} 	&  \textbf{5.06} 	& \textbf{0.996} & 1.11  & 5.52 \\
                        200	&	3	& 85.8 	& 14.2 	& 0.686 & 0.03  & 3.83 \\
\rowcolor{gray!25}    199	&	4	& \textbf{94.9} 	&  \textbf{5.09} 	& \textbf{0.995} & 0.15  & 4.54 \\
                        -	&	4	& - 	& - 	& -  	& -  	& - \\
\rowcolor{gray!25}   263	&	5	& \textbf{95.0} &  \textbf{5.04}	& \textbf{0.998} & 0.02  & 3.96 \\
                        315	&	5	& 90.1 	& 9.88 	& 0.833 & 0.00 	& 3.31 \\
\end{tabular} 
\caption{Performance of the single-sampling plans $(n,c)$ whose OCs are shown in Fig.~\ref{fig_OC_opt_MID_ISO_matrix}. 
The producer's (consumer's) risk quality PRQ (CRQ) $\qa$ ($\qb$) is the quality level where the allowed bound is reached, namely $\Pac(\qa)=\Pa=95\%$ ($\Pac(\qb)=\Pb=5\%$).  
Gray-colored rows are the MID-optimised sample plans, and white rows are the standard members of the ISO 2859-1 family; these offer no sample size for $c=4$. 
Bold figures are close to, i.e., within $1\%$ [$0.1\%$ for the PRQ $\qa$] of the MID conditions.  The least biased, minimal sample plan is (88,2).
} 
\label{tab_nc_opt_MID_ISO_binom}
\end{table*}  

Figure \ref{fig_OC_opt_MID_ISO_matrix} shows the resulting OCs for various single-sampling plans 
$(n,c)$, grouped with increasing acceptance number $c=0,1,2,\dots$ into pairs. 
The dashed curve shows the respective member of the ISO 2859-1 family at inspection level II as recommended in \cite{Welmec8.10}.%
\footnote{For $c=4$, ISO 2859-1 offers no sample size. Instead, it jumps directly from $(200,3)$ to $(315,5)$.}
The full curve shows the smallest sample compatible with both MID conditions, which can be computed using standard numerical tools or commercial software, lowering the sample size until one of the two MID conditions is violated. Not surprisingly, 
when relaxing the constraint to the somewhat arbitrary members of the ISO 2859-1 family, one obtains considerably smaller sample sizes. 
\rev{%
Qualitatively, one arrives at the same conclusion when approximating the binomial distribution \eqref{binomialcgeneral} for $p\to 0$, $n \to \infty$ at fixed $pn$  by the cumulative Poisson distribution $\Pac(p,n,c) \approx \sum_{k=0}^c
e^{-np}(np)^k/k!$. 
Indeed, the respective minimal sample sizes under the Poisson approximation are, for $c=0$, $n=43(+1)$, for $c=1$, $n=68 (+2)$, for $c=2$, $n=90 (+2)$, for $c=3$, $n=137(-1)$, for $c=4$, $n=198 (-1)$, for $c=5$, $n=262(-1)$, etc. 
Here, each number in parentheses denotes the difference to the optimal, and more accurate binomial result displayed in the first two columns and grey-colored rows of Table \ref{tab_nc_opt_MID_ISO_binom}.   
}

Because the probabilities for the different cases $k=0,1,2,\dots,c$ in  Eq.~\eqref{binomialcgeneral} add up and still must stay below the MID points, the minimum sample size grows with the acceptance number.
In return, the OCs gain in specificity (smaller $\alpha$) and sensitivity or statistical power (smaller $\beta$), i.e., permit to distinguish more accurately between high and low quality levels.    
For small acceptance numbers $c=0,1$, the LQ condition of \MIDb{} is the more stringent, i.e., OCs are not steep enough. For large acceptance numbers $c>2$, the AQL condition of \MIDa{} becomes more stringent, i.e., OCs are steeper than required. 
The optimal sample plan, with smallest sample size $n$ while least biased with respect to both MID conditions, is found to be $(n,c)=(88,2)$.

Table \ref{tab_nc_opt_MID_ISO_binom} lists the corresponding data, allowing for a quantitative comparison of the ISO 2859-1 and MID-optimised sample plans. 
As Fig.~\ref{fig_OC_opt_MID_ISO_matrix} already shows, the LQ criterion \MIDb{} can be saturated quite well for low acceptance numbers, with a consumer's risk quality (CRQ) $\qb$ such that $\Pac(\qb)=\Pb$ not much below $\pb=7\%$, and a consumer's risk $\beta$ not much below $\Pb=5\%$. 
The price to be paid for small acceptance numbers and sample sizes is an elevated producer's risk $\alpha$, i.e., a high probability for the type I error of rejecting a good lot. 
And the smallest admissible member $(n,c)=(50,0)$ of the ISO 2859-1 family at inspection level II is stricter than necessary with a producer's risk of $\alpha=39.5\%$; the corresponding producer's risk quality (PRQ) required to reach $\Pac(\qa)=\Pa=0.95$ is as low as $\qa=0.103\%$. 
The minimal single-sampling plan $(n,c)=(42,0)$ compatible with MID conditions implies a slightly smaller producer's risk $\alpha \approx 34.4\%$ with a slightly larger PRQ of $\qa=0.122$ and a CRQ $\qb$ just below $7\%$.  

Conversely, for larger acceptance numbers $c\geq 2$ and, thus, larger sample sizes, the AQL criterion \MIDa{} becomes saturated with a producer's risk $\alpha$ approaching $5\%$ from above. Here, the consumer's risk drops to values $\beta\ll 1\%$ much lower than required by MID, and a CRQ $\qb$ substantially smaller than $7\%$. 
Larger samples are indeed generally known to reduce the probabilities of type I and II errors and to have a greater discriminating power 
\cite{SchillingNeubauer2017,Mathews2010}.  
Thus, larger samples and higher acceptance numbers have their merits in internal production control and may well be suggested in the corresponding MID modules.
However, a systematic growth of sample size with lot size, as recommended by the ISO 2859-1 sampling system, is not warranted by the MID conditions for modules F and F1.    

In summary so far: 
For large enough lot sizes $N\gg n$ (see Section~\ref{sec:summary_finiteN} for a quantitative discussion) 
the minimal, least biased sample plan for statistical product verification in the MID modules F and F1 is $(n,c)=(88,2)$.

\section{Finite lot sizes} 
\label{finiteN.sec}

\begin{figure*}
\centering
\includegraphics[width=.9\textwidth]{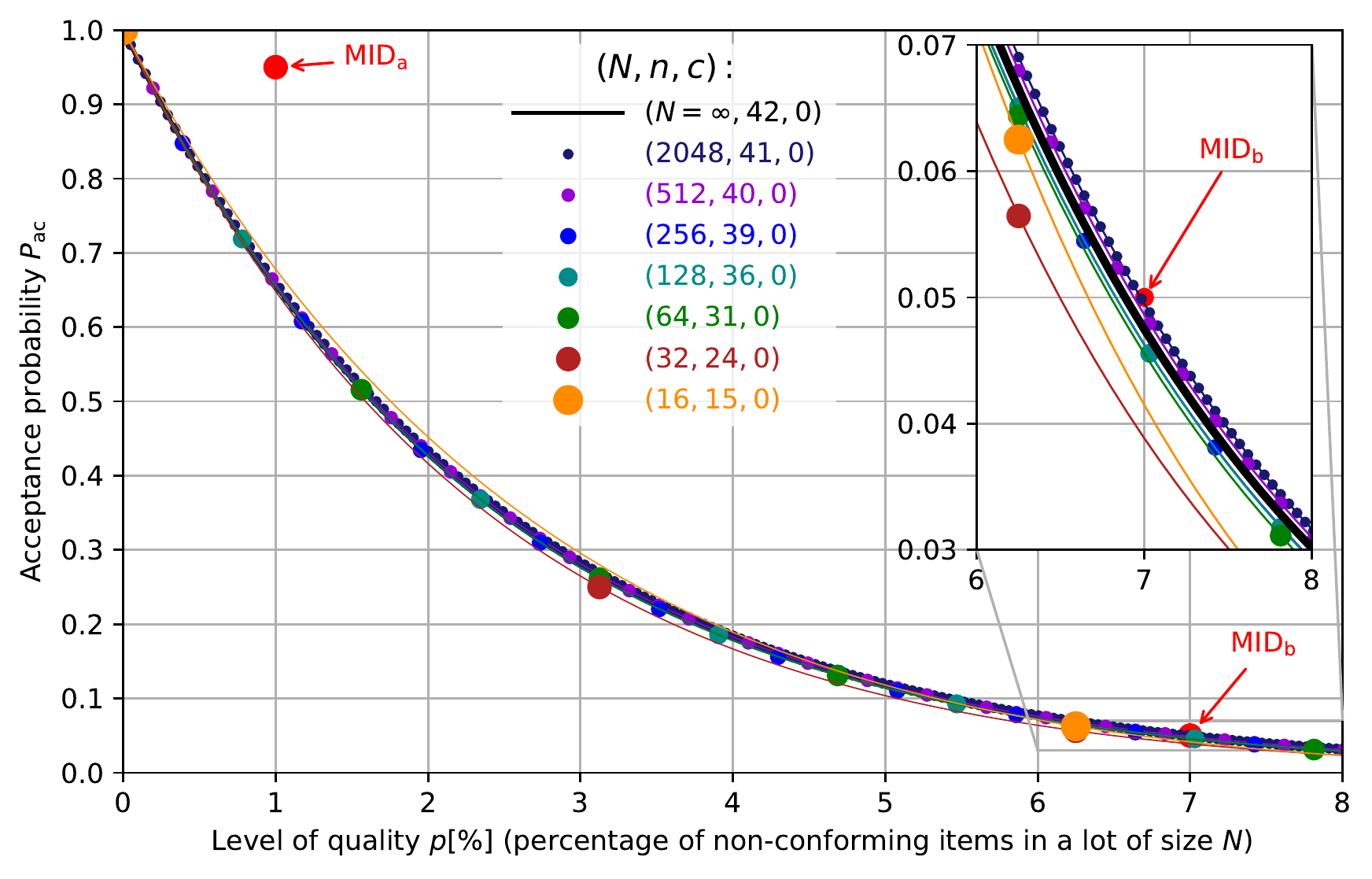}
\caption{Acceptance probability $\Pac(pN,N,n,0)$ of minimal zero-acceptance single sampling plans admissible under MID conditions, for lot sizes $N$ following a geometric progression, plotted as function of the lot's quality level or fraction defective $p=M/N$.    
Filled colored dots correspond to the hypergeometric distribution \eqref{hypergeomc0}; connecting curves are guides to the eye by the extension  \eqref{analyticcont}. The full black line is the binomial sampling model for $(n=42,c=0)$ reached in the limit $N\to\infty$. 
} 
\label{figN_c0_OCs_Nn}  
\end{figure*}

\begin{figure*}
\centering
\includegraphics[width=.9\textwidth]{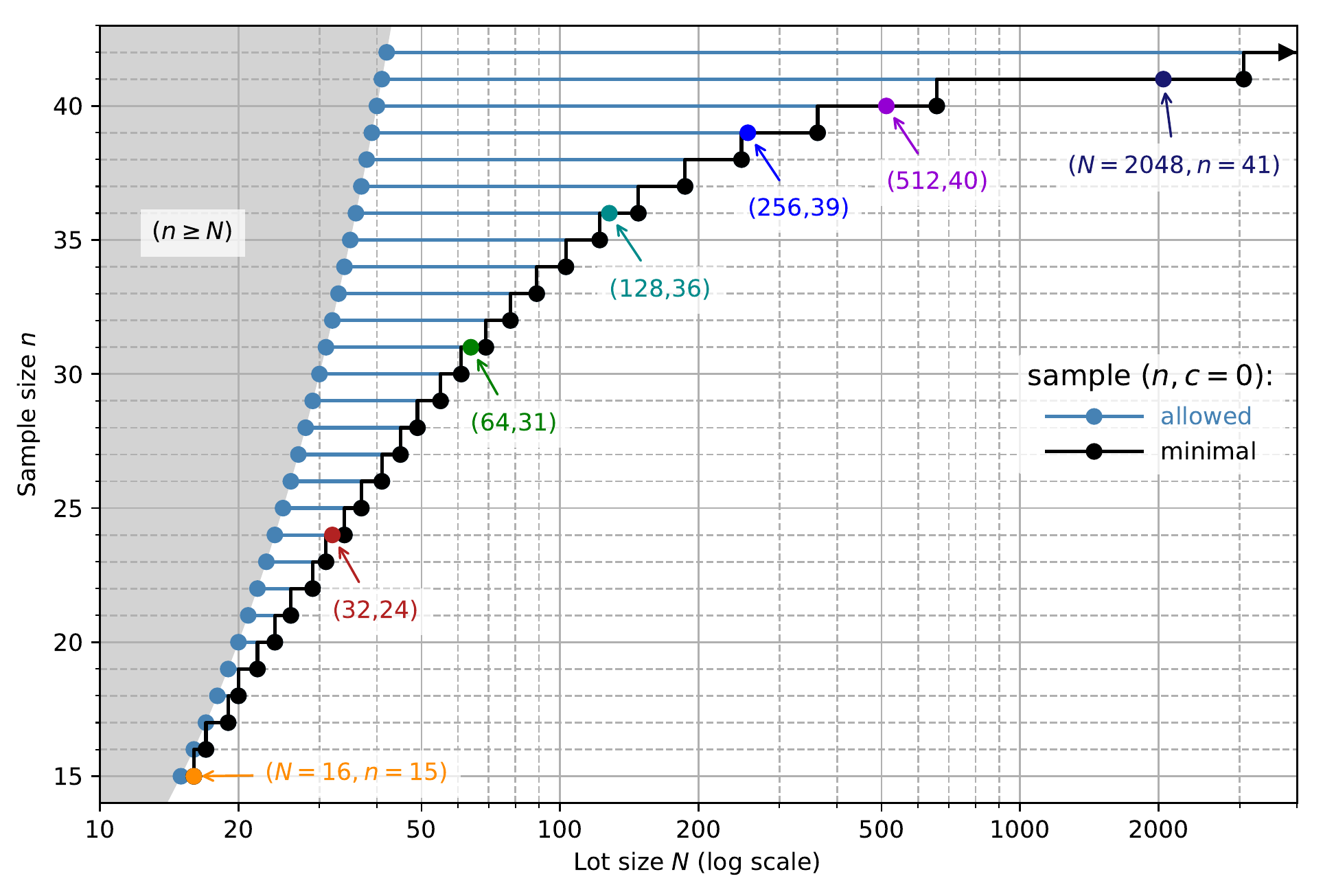}
\caption{%
MID-compatible sample size $n$ for zero-acceptance single sampling as function of lot size $N$. 
For $1\leq N\leq 15$, only 100\%  testing $n=N$ is admissible (see text). Between $N=16$ and $N=3063$, the minimum sample size grows slowly from $n=15$ to $n=41$, as listed in Table \ref{tab_nmin_N_c0}. 
From $N\geq 3064$ onward, the mimimum sample size saturates at $n=42$ already known from the binomial model (compare Fig.~\ref{fig_OC_opt_MID_ISO_matrix}). Colored dots correspond to the OCs plotted in Fig.~\ref{figN_c0_OCs_Nn}.
} 
\label{fig_n_N1N1_c0}  
\end{figure*} 

If the lot size is not much larger than the sample size then finite-size corrections become noticeable, and the results of the previous section have to be revisited \cite{SchillingNeubauer2017}. Let us consider a sample of $n$ items drawn without replacement from a lot of size $N$ containing $M\in\{0,1,\dots,N\}$ defective items.  
Under the single-sampling paradigm, the lot is to be accepted if the sample contains at most $c$ defective items. The acceptance probability is then given by the cumulative density of the hypergeometric distribution,
\be
\Pac(M,N,n,c) = \sum_{k=0}^c \frac{\binom{M}{k}\binom{N-M}{n-k}}{\binom{N}{n}},  
\label{hypergeom} 
\ee 
where each summand is the combinatorial probability to find exactly $k=0,1,\dots,c$ defective items among the $n$ items tested. 
One can assign the quality level $p=M/N$ to this lot and discuss its acceptance probability $\Pac(pN,N,n,c)$ for fixed $N$ as function of the operationally meaningful values $p\in \mathcal{D}_N$, where 
\be \label{DN.def}
\mathcal{D}_N=\{0,\tfrac{1}{N},\tfrac{2}{N},\dots,\tfrac{N-1}{N},1\}.   
\ee
is the domain of the function $\Pac(\cdot,N,n,c):p\in\mathcal{D}_N\mapsto\Pac(p,N,n,c)\in\Pac(\mathcal{D}_N)$.

From this general setting follows
\begin{prop} 
A single-sampling plan with acceptance number $c$ can only be admissible under the MID condition \eqref{MID3} for lot sizes 
\be\label{lowerboundN}
 N > 100 c. 
\ee
\end{prop} 
\begin{proof} 
A lot containing at most $M=c$ defective items will certainly be accepted by any sample plan $(n,c)$. But then the corresponding quality level $p=\frac{c}{N}$ (where $\Pac=1$) must be smaller than the AQL $\pa$ since otherwise condition \eqref{MID3} cannot be satisfied. Therefore, $\frac{c}{N} < \pa=\frac{1}{100}$, which is equivalent to \eqref{lowerboundN}. 
\end{proof}

\begin{table*}\centering
\setlength{\tabcolsep}{0pt}
\rowcolors{3}{white}{gray!25}
\begin{tabular}{*{2}{C{2.75em}}|*{2}{C{2.25em}}|*{2}{C{2cm}}|*{2}{C{2cm}}}
\multicolumn{2}{c|}{Lot size $N$}	& 
    \multicolumn{2}{c|}{Sample}& 	
    \multicolumn{2}{c|}{Producer's risk $\alpha$ [\%]} & 
    \multicolumn{2}{c}{Consumer's risk	$\beta $ [\%]}\\
from & to 	&$n$	&	$c$ & from & to& from & to \\
\hline  
  15 &   16 &   15 &   0 &     (40.37) &    (32.21) &     0.00 &     4.15\\ 
  17 &   17 &   16 &   0 &     (34.53) &    (34.53) &     3.00 &     3.00\\ 
  18 &   19 &   17 &   0 &     (36.82) &    (32.66) &     2.13 &     4.46\\ 
  20 &   20 &   18 &   0 &     (34.73) &    (34.73) &     3.42 &     3.42\\ 
  21 &   22 &   19 &   0 &     (36.77) &    (33.94) &     2.61 &     4.05\\ 
  23 &   24 &   20 &   0 &     (35.84) &    (33.74) &     3.19 &     4.31\\ 
  $\vdots$&$\vdots$&$\vdots$&$\vdots$&$\vdots$&$\vdots$&$\vdots$&$\vdots$\\
  56 &   61 &   30 &   0 &     (34.76) &    (33.70) &     4.36 &     4.94\\ 
  62 &   69 &   31 &   0 &     (34.81) &    (33.66) &     4.36 &     4.99\\ 
  70 &   78 &   32 &   0 &     (34.71) &    (33.71) &     4.43 &     4.99\\ 
  79 &   89 &   33 &   0 &     (34.72) &    (33.77) &     4.45 &     4.98\\ 
  90 &  103 &   34 &   0 &     (34.73) &    33.81 &     4.47 &     4.97\\ 
 104 &  122 &   35 &   0 &     34.74 &    33.83 &     4.48 &     4.98\\ 
 123 &  148 &   36 &   0 &     34.71 &    33.85 &     4.51 &     4.99\\ 
 149 &  187 &   37 &   0 &     34.70 &    33.86 &     4.54 &     5.00\\ 
 188 &  248 &   38 &   0 &     34.66 &    33.89 &     4.57 &     5.00\\ 
 249 &  363 &   39 &   0 &     34.66 &    33.91 &     4.59 &     5.00\\ 
 364 &  659 &   40 &   0 &     34.65 &    33.93 &     4.61 &     5.00\\ 
 660 & 3063 &   41 &   0 &     34.63 &    33.95 &     4.63 &     5.00\\ 
3064 & $\infty$ &   42 &   0 &     34.62 &    34.4 &     4.65 &     4.75\\ 
\end{tabular} 
\caption{%
Minimum sample size $n$ allowed by the MID for zero-acceptance single sampling as function of lot size $N$, as shown in Fig.~\ref{fig_n_N1N1_c0}, together with the producer's and consumer's risk. 
Figures in parentheses correspond to lot sizes $N<100$, where a single defective item already defines a quality level $p=1/N > \pa$ such that the producer's risk $\alpha=1-\Pac(\pa)$ has no operational meaning. 
Since the sample plans are optimised with respect to the MID conditions, the risk data stray only slightly from the values obtained in the limit $N=\infty$, taken from the binomial model of Table \ref{tab_nc_opt_MID_ISO_binom}.
Entries ``$5.00$'' arise due to rounding from the true values, which obey the strict inequality \eqref{MID4}. 
} 
\label{tab_nmin_N_c0}
\end{table*}  

Analogously to the optimisation of sample size in Section \ref{sec:Ninfty}, one can further determine the smallest sample size $n$, given acceptance number $c$ and lot size $N>100c$, that is admissible under the MID conditions. 
We choose to use the criteria \eqref{MID3} and \eqref{MID4}, namely checking whether the acceptance probability \eqref{hypergeom} 
at $M_\text{a}=\pa N$ and $M_\text{b} = \pb N$  is inferior to the bounds $\Pa$ and $\Pb$, respectively.  
This requires the extension of the factorials in \eqref{hypergeom} to non-integer arguments, which is easily achieved using the gamma function \cite{AbramovitzStegun}: 
\be\label{analyticcont}
x!=\Gamma(x+1)=\int_0^\infty t^x e^{-t} dt. 
\ee
While the notion of non-integer values for $M_\text{a}=\pa N$ and $M_\text{b} = \pb N$ is not evident to justify operationally for single lots, this procedure turns out to be more consistent than a purely discrete formulation; for a detailed justification see Appendix \ref{app:MID_finiteN}.    

We find that the manner in which finite lot size affects MID-optimised sample plans depends crucially on the acceptance number. 
Details for the most relevant cases $c=0,1,2$ are discussed in the following subsections \ref{finiteN_c0.sec} through \ref{finiteN_c2.sec}. 
Readers mainly interested in the final results are invited to consult Sec.~\ref{sec:summary_finiteN} straight away.

\subsection{Zero acceptance $c=0$}
\label{finiteN_c0.sec}

For $c=0$, the binomial prediction for very large lots ($N=\infty$) is conservative in the sense that it overestimates the sample size that is really necessary for a lot of a certain finite size $N$ \cite{SchillingNeubauer2017}. 
Thus, $(50,0)$ from ISO 2859-1 has been correctly identified as being compatible with MID for lot sizes from $51\leq N\leq 500$ \cite{Welmec8.10}.  
Similarly, the single-sampling plan $(42,0)$ is certainly admissible, all the way down to $N=43$. 
However, taking into account finite lot sizes allows us to reduce the required sample sizes even further. 

Let us first explain qualitatively why smaller lots require smaller zero-accceptance samples. 
In the binomial model, the probability $p$ to draw defective items without replacement from the lot is taken constant. However, if a lot of size $N$ contains $M$ defective items, then the fraction defective is  $p=M/N$ only for the first draw.  
The probabilities of the second draw depend on the outcome of the first. 
If the first item is defective, then the second item is defective with probability $p'=(M-1)/(N-1)$, which is smaller than $p$ (we can assume $M<N$, since for $M=N$ one has $p=p'=1$, a trivial case without practical interest because all lots are rejected anyway). 
Vice versa, if the first item is conforming, then $p'=M/(N-1)$, which is larger than $p$ (we assume $M>0$, otherwise we have $p=p'=0$, again a trivial case where all lots are accepted). 
The latter case $p'>p$ is more frequent since $p\ll 1$ in realistic settings. 
This evolution of the probability to larger values will likely continue for each draw, so that the actual chance to discover defective items in a lot of small size is larger than predicted by the binomial model. 
By consequence, the actual acceptance probability for the same sample size would be smaller, such that a smaller sample actually suffices to stay below the MID bounds.

For $c=0$, the general expression \eqref{hypergeom} for the acceptance probability simplifies somewhat, 
\be
\Pac(M,N,n,0) = \frac{(N-M)!(N-n)!}{N!(N-M-n)!}. 
\label{hypergeomc0} 
\ee
With the help of the factorial extension \eqref{analyticcont}, one can determine numerically the smallest sample size $n$, given the lot size $N$, that still fulfills both MID conditions. It turns out that 
actually only the LQ condition \eqref{MID4} at \MIDb{} matters.
Figure \ref{figN_c0_OCs_Nn} shows the result of such a minimisation, for lot sizes $N$ growing in geometric progression toward the limit $N=\infty$ where the binomial result $n=42$  of Section~\ref{sec:Ninfty} becomes exact. 
It is evident that sample sizes can be substantially reduced compared to the binomial model for small to moderate lot sizes.

In theory, the smallest lot size for which statistical sampling can be envisaged under the MID conditions is $N=15$. 
The reason is that for $1\leq N\leq 14$, the quality level of a lot with a single defective item is at least $1/14\approx7.143\%$ and thus already larger than $\pb=7\%$ that should be rejected. 
Therefore,  the only way to ensure MID conditions for $N\leq14$ is 100\% testing with $n=N$. 
The next larger lot size $N=15$ is such that a sample size of $n=14$ violates condition \eqref{MID4}.  
Thus, also $N=15$ requires 100\%  testing with $n=15$. 
The combination $N=16$ and $n=15$, however, is compatible with the MID conditions, as shown by the corresponding OC in Fig.~\ref{figN_c0_OCs_Nn}. 
For $N=17$, one has to step up to $n=16$, and so on, 
all the way up to $N=3063$ and $n=41$. From $N\geq 3064$ onward, the required sample size saturates at $n=42$ already derived from the binomial model (cf.~Fig.~\ref{fig_OC_opt_MID_ISO_matrix}). Figure \ref{fig_n_N1N1_c0} shows the allowed and minimum sample size $n$ found under MID conditions as function of the lot size $N$. 

The corresponding lot-size intervals with their minimum sample sizes are listed in Table \ref{tab_nmin_N_c0}, together with the producer's and consumer's risks. Since the sample size is optimised with respect to the MID conditions (here, for $c=0$, only the consumer's point \MIDb{} matters), these data vary only slightly from one case to the other. The main message is that sample size can be substantially reduced, under very similar risks, for finite lot sizes all the way down to $N=16$.

\subsection{Unit acceptance $c=1$}
\label{finiteN_c1.sec}

\begin{figure*}\centering
\includegraphics[width=.9\textwidth]{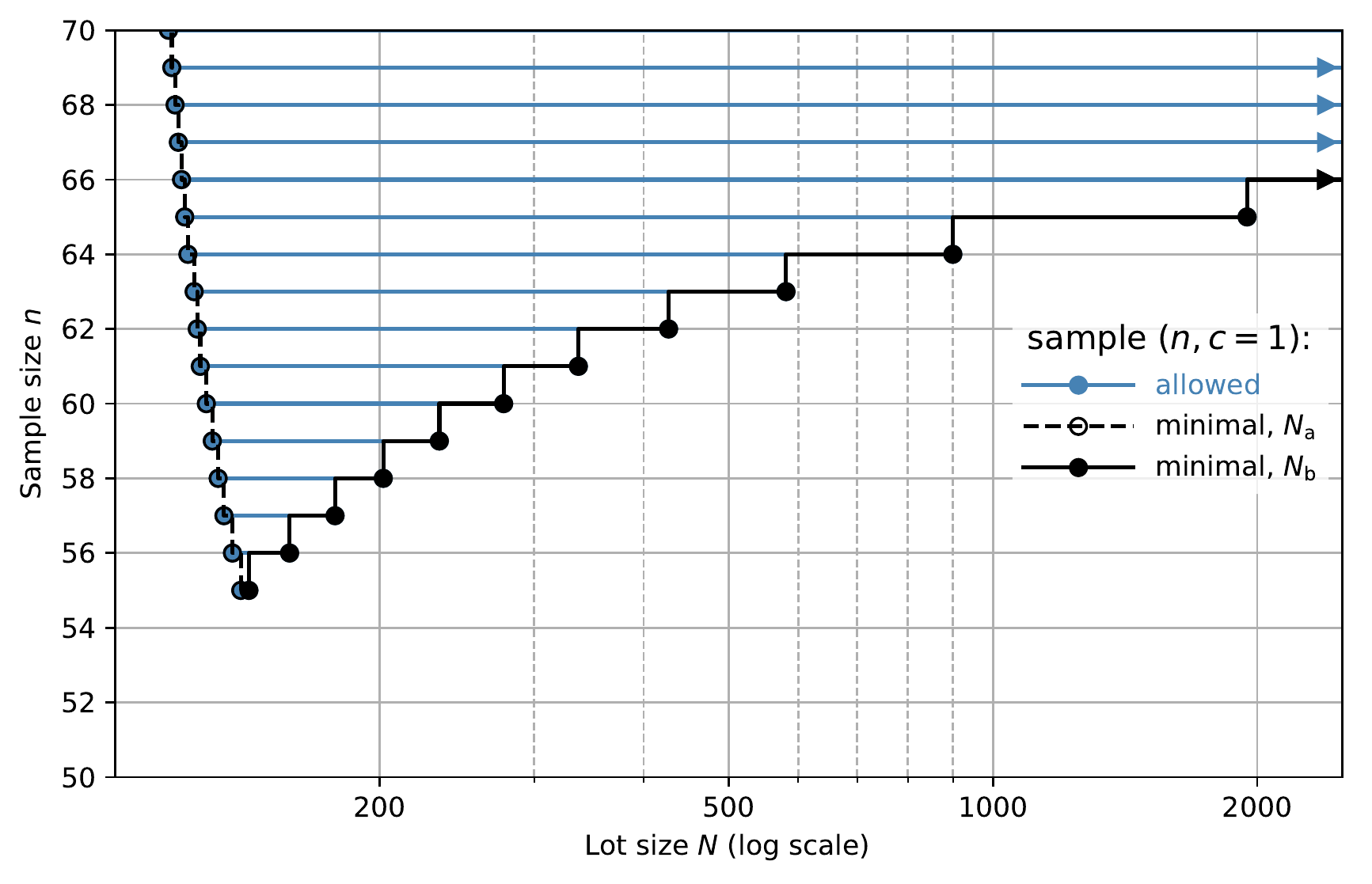}
\caption{%
MID-allowed sample size $n$ for unit-acceptance $(c=1)$ single sampling as function of lot size $N$. 
The smallest (largest) allowed lot size $N_\text{a}$ ($N_\text{b}$) for each sample size $n$ is determined by the point \MIDa{} (\MIDb{}). 
Black lines and dots indicate minimum sample sizes on the `a' and `b' side, respectively. 
}   
\label{fig_n_N1N2_c1}  
\end{figure*}

\begin{table*}\centering
\setlength{\tabcolsep}{0pt}
\rowcolors{3}{gray!25}{white}
\begin{tabular}{*{2}{C{4em}}|*{2}{C{2.25em}}|*{2}{C{2cm}}|*{2}{C{2cm}}}
\multicolumn{2}{c|}{Lot size $N$}	& 
    \multicolumn{2}{c|}{Sample}& 	
    \multicolumn{2}{c|}{Producer's risk $\alpha$ [\%]} & 
    \multicolumn{2}{c}{Consumer's risk	$\beta $ [\%]}\\
from $N_\text{a}$ & to $N_\text{b}$ 	&$n$	&	$c$ & from & to& from & to \\
\hline 
139 &  142 &   55 &   1 &      5.07 &     5.26 &     4.88 &     4.98\\ 
 136 &  158 &   56 &   1 &      5.05 &     6.37 &     4.33 &     4.98\\ 
 133 &  178 &   57 &   1 &      5.02 &     7.41 &     3.83 &     4.99\\ 
 131 &  202 &   58 &   1 &      5.04 &     8.35 &     3.39 &     4.99\\ 
 129 &  234 &   59 &   1 &      5.05 &     9.26 &     2.98 &     4.99\\ 
 127 &  277 &   60 &   1 &      5.04 &    10.10 &     2.61 &     5.00\\ 
 125 &  337 &   61 &   1 &      5.00 &    10.88 &     2.27 &     5.00\\ 
 124 &  427 &   62 &   1 &      5.08 &    11.62 &     1.99 &     5.00\\ 
 123 &  581 &   63 &   1 &      5.14 &    12.31 &     1.74 &     5.00\\ 
 121 &  900 &   64 &   1 &      5.05 &    12.97 &     1.48 &     5.00\\ 
 120 & 1947 &   65 &   1 &      5.09 &    13.59 &     1.28 &     5.00\\ 
 119 & $\infty$ &   66 &   1 &      5.12 &    14.1   &    1.10 &     4.96\\ 
 \end{tabular}                   
\caption{%
Sample size $n$ for unit-acceptance $(c=1)$ single sampling in the lot-size intervals $[N_\text{a},N_\text{b}]$ allowed by MID, together with the producer's and consumer's risk. 
At $N_\text{a}$, the producer's risk is optimally close to (and just above) $\Pa=5\%$; conversely, at $N_\text{b}$, the consumer's risk is optimally close to (and just below) $\Pb=5\%$. For $N=\infty$ the data is taken from the binomial model (cf.~Table \ref{tab_nc_opt_MID_ISO_binom}).
} 
\label{tab_n_N1N2_c1}
\end{table*}  

Just as for $c=0$, one can minimize the sample size $n$ for a given $N$ with fixed $c=1$ such that the MID conditions are fulfilled. 
In contrast to the case $c=0$, now also the producer's point \MIDa{} plays a significant role, which brings about a qualitative difference in the behavior of the OCs as function of system size $N$. 
Indeed, Eq.~\eqref{lowerboundN} of Proposition 1 tells us already that 
the lot size is globally bounded from below by $N>100$.
The precise value of this lower bound depends on the acceptance probability at the AQL point \MIDa{} and thus on the sample size $n$.   
Conversely, the lot size is bounded from above, at fixed sample size, by the behaviour at the LQ point \MIDb{}.  
Figure \ref{fig_n_N1N2_c1} shows the resulting, allowed combinations of sample size $n$ and lot sizes $N\in[N_\text{a},N_\text{b}]$.  
No MID-compatible sampling with $n<55$ is possible.
As shown by the `a'  branch on the left side, the \MIDa{} point requires the minimum sample size to grow sharply when the sample size decreases toward the absolute lower bound \eqref{lowerboundN} at $N=100$.  For all larger lot sizes $N\geq 1948$, the sampling $(66,1)$ already known from the binomial model becomes optimal.

Table \ref{tab_n_N1N2_c1} shows the quantitative risk data associated with these intervals.  By construction, for the smallest allowed lot size $N_\text{a}$ in each row (the ``from'' case), the producer's risk is optimally close to (and just above) the limit 5\% imposed by the AQL point \MIDa{}; conversely, for the largest allowed lot size (the ``to'' case), the consumer's risk is optimally close to (and just below) the 5\% limit imposed by the LQ point \MIDb{}.

\subsection{Double acceptance $c=2$ }
\label{finiteN_c2.sec}

\begin{figure*}\centering
\includegraphics[width=0.9\textwidth]{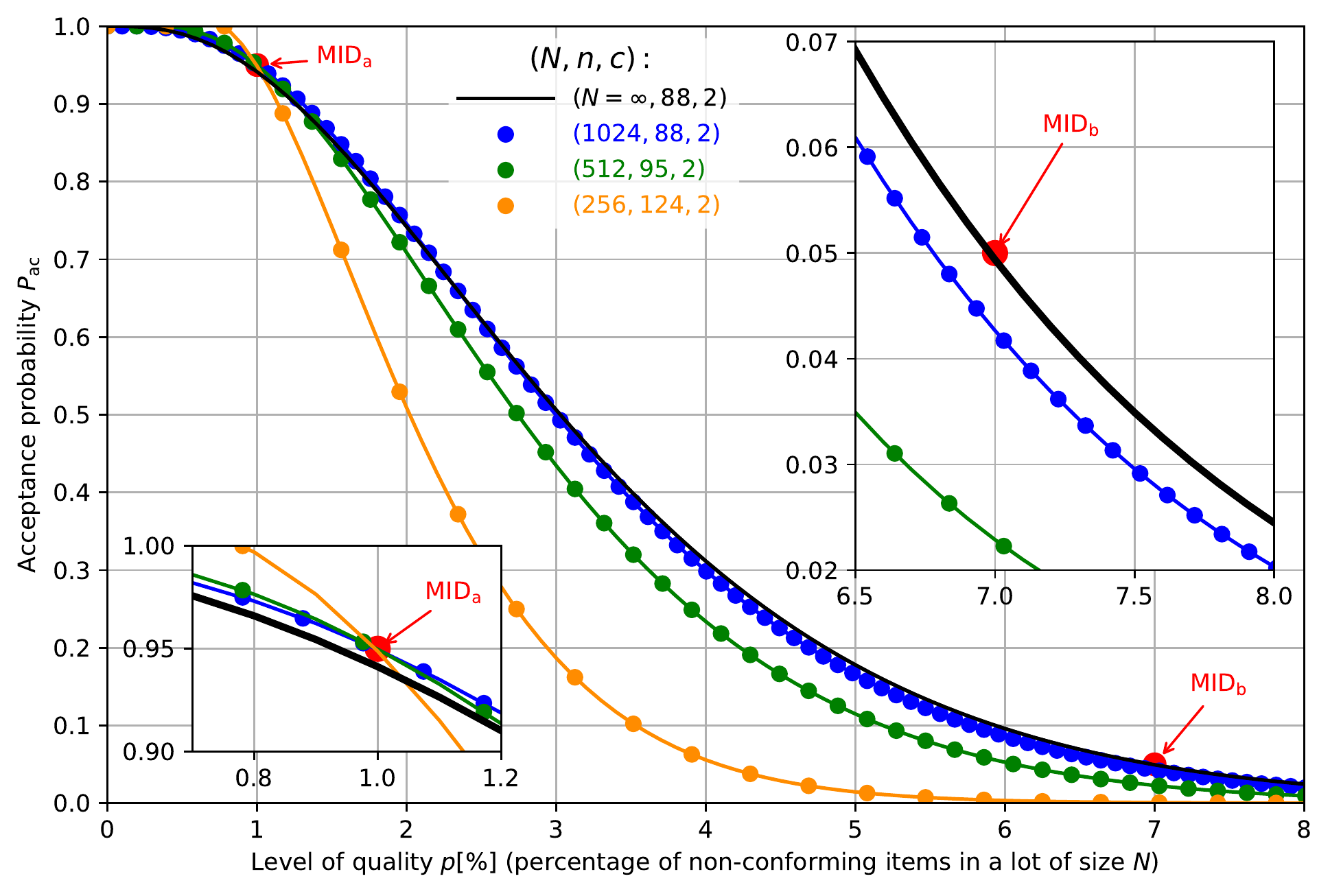}
\caption{%
Operating characteristics for lot sizes $N=256,512,1024$ and their minimal double-acceptance $(c=2)$ single sampling plans under MID conditions. 
 As shown in the inset on the lower left, now the AQL condition \MIDa{} becomes increasingly hard to satisfy for smaller lots, such that the required sample size has to increase quite dramatically. The exaggerated steepness of the OC curves for smaller lots indicates that in those cases sampling with a lower acceptance number ($c\leq1$) is more appropriate.  
} 
\label{fig10}  
\end{figure*}

For $c=2$, now the AQL point \MIDa{} mainly determines the allowed combinations of lot size $N$ and sample size $n$. For an illustration, Figure \ref{fig10} shows the OCs of 3 different sample sizes $N=256,512,1024$ with their minimum sample sizes $n=124,95,88$ determined by the \MIDa{} condition. 
Since the binomial limit $(88,2)$ is conservative in the sense that the acceptance probability \eqref{hypergeom} at $\pa$ \emph{increases} as lot size $N$ decreases, then the minimum sample size has to increase as well for smaller lots in order to compensate this effect. 
This one-sided constraint makes the \MIDb{} point increasingly irrelevant as $N$ becomes smaller, as evident from the OC for the smallest lot size $N=256$ plotted in Fig.~\ref{fig10}. 
The rapid decrease of the OC indicates that the acceptance number $c=2$ is too elevated for such a small lot size, suggesting instead to fall back onto  $c=1$ or even  $c=0$.

\begin{figure*}\centering
\includegraphics[width=0.9\textwidth]{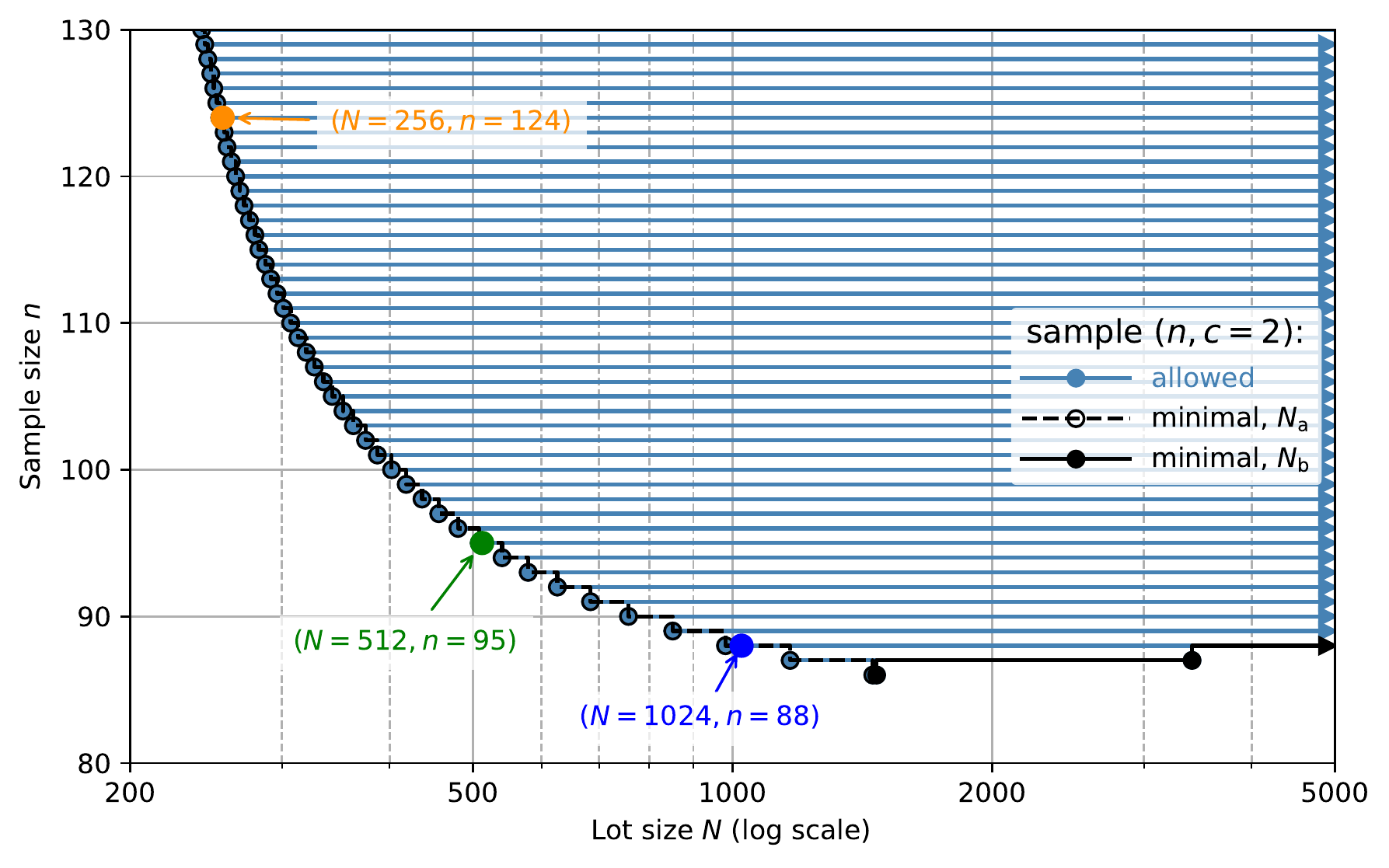}
\caption{%
MID-allowed sample size $n$ for double-acceptance $(c=2)$ single sampling as function of lot size $N$.  
The smallest (largest) allowed lot size $N_\text{a}$ ($N_\text{b}$) for each sample size $n$ is determined by the point \MIDa{} (\MIDb{}). Here the main factor is the \MIDa{} condition, only the samples $n=86,87$ are limited to finite lot sizes by the \MIDb{} condition. The binomial result $(88,2)$ remains valid for arbitrarily large lots.  
The 3 colored dots correspond to the OCs plotted in Fig.~\ref{fig10}.
}   
\label{fig_n_N1N2_c2}  
\end{figure*} 

Figure \ref{fig_n_N1N2_c2} shows the allowed combinations of lot and sample sizes. The minimum sample size found from the binomial model, $n=88$, is admissble only down to $N=981$. 
For smaller lots, the minimum sample size has to increase quite dramatically in order to satisfy the \MIDa{} condition. The \MIDb{} constraint only has a marginal influence: it limits the validity of the smallest possible sample sizes $n=86,87$ to $N_\text{b}=1469, 3412$, respectively.  
The optimal binomial result $(n=88,c=2)$ is valid up to arbitrarily large lots, of course.

\begin{table*}\centering
\setlength{\tabcolsep}{0pt}
\rowcolors{3}{gray!25}{white}
\begin{tabular}{*{2}{C{4em}}|*{2}{C{2.25em}}|*{2}{C{2cm}}|*{2}{C{2cm}}}
\multicolumn{2}{c|}{Lot size $N$}	& 
    \multicolumn{2}{c|}{Sample}& 	
    \multicolumn{2}{c|}{Producer's risk $\alpha$ [\%]} & 
    \multicolumn{2}{c}{Consumer's risk	$\beta $ [\%]}\\
from $N_\text{a}$ & to 	 $N_\text{b}$ &$n$	&	$c$ & from & to& from & to \\
\hline 
1454 & 1469      & 86   & 2   &      5.00 &     5.01 &     4.99 &       5.00\\
1166 & 3412      & 87   & 2   &      5.00 &     5.48 &     4.60 &       5.00\\
981 &  $\infty$  &   88 &   2 &      5.00 &     5.87 &     4.24 &       4.94\\ 
 852  &  $\infty$&   89 &   2 &      5.00 &     6.03 &     3.90 &       4.68\\ 
 757  &  $\infty$&   90 &   2 &      5.00 &     6.19 &     3.58 &       4.44\\ 
 684  &  $\infty$&   91 &   2 &      5.00 &     6.36 &     3.28 &       4.21\\ 
 626  &  $\infty$&   92 &   2 &      5.00 &     6.53 &     3.00 &       3.99\\ 
 579  &  $\infty$&   93 &   2 &      5.00 &     6.70 &     2.74 &       3.78\\ 
 540  &  $\infty$&   94 &   2 &      5.00 &     6.87 &     2.50 &       3.58\\ 
 508  &  $\infty$&   95 &   2 &      5.00 &     7.04 &     2.28 &       3.39\\ 
 480  &  $\infty$&   96 &   2 &      5.00 &     7.22 &     2.07 &       3.21\\ 
 $\vdots$ &&&&&&&\\
 257  &  $\infty$&  123 &   2 &     5.03  &     12.62     &0.08  &   0.70\\ 
 254  &  $\infty$&  124 &   2 &      5.00 &    12.84 &     0.07 &     0.66\\ 
 252  &  $\infty$&  125 &   2 &      5.02 &    13.07 &     0.06 &     0.62\\ 
 $\vdots$ &&&&&&&
 \end{tabular}  
\caption{%
MID-compatible lot and sample sizes for double-acceptance 
($c=2$) single sampling, together with producer's and consumer's risk. 
The minimum sample size $n=88$ found within the binomial model is also the minimum sample size found here, valid for  $N\geq981$. 
Smaller lots require larger samples.
For $N<500$, the sample size rises substantially while the consumer's risk drops far below the MID bound of $5\%$, indicating that unit-acceptance ($c=1$) sampling is more appropriate.          
} 
\label{tab_n_N1N2_c2}
\end{table*}

Table \ref{tab_n_N1N2_c2} lists the allowed sample size as function of the lot size intervals, together with the producer's and consumer's risks. 
It is clear that for smaller lots, double-acceptance sampling is no longer the best choice for a fair realization of the MID criteria. 
For example, a lot of size $N=256$ requires the minimum double-acceptance sampling $n=124$ and features only a consumer's risk of $\beta=0.07\%$, two orders of magnitude smaller than required by \MIDb{}. 
An arguably better choice then is $c=1$, for which it can be found in Table \ref{tab_n_N1N2_c1} that the smallest allowed unit-acceptance sampling $(60,1)$ has  $\alpha\approx 10\%$ and $\beta \approx 5\%$.

\subsection{Summary of results for finite lot sizes} 
\label{sec:summary_finiteN}

\begin{figure*}\centering
\includegraphics[width=0.9\textwidth]{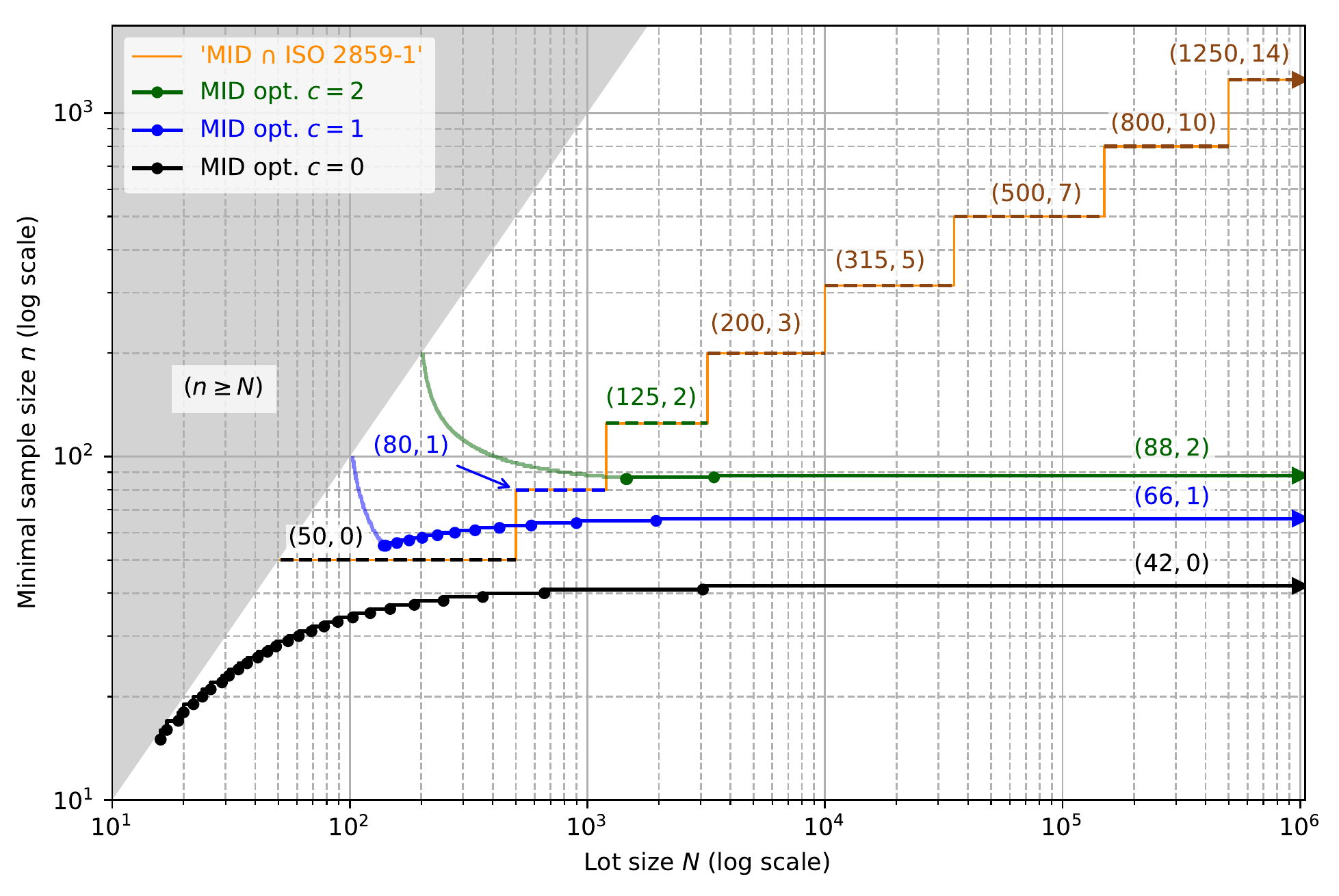}
\caption{
Combined view of MID-minimised single-sampling size $n$ as function of lot size $N$ for acceptance numbers $c=0,1,2$, together with the plans of ISO 2859-1  recommended by \cite{Welmec8.10} (steps), on a double-logarithmic scale. 
}   
\label{fig_nmin_cALL_withISO}  
\end{figure*}

Figure \ref{fig_nmin_cALL_withISO} summarizes the impact of finite lot sizes on the MID-optimised sampling scheme with acceptance number $c=0,1,2$, covering several orders of magnitude on a double logarithmic scale. The main results are: 
\begin{enumerate} 
\item The minimal sample plans $(n,c)$ obtained within the binomial model for large lots in Sec.~\ref{sec:Ninfty}, $(42,0)$, $(66,1)$, and $ (88,2)$, are indeed the minimal sample plans for all lot sizes $N > N_c$ where $N_0=3063$, $N_1=1947$, and $N_2= 3412$. 
Finite-lot-size corrections are of two different types, dictated by the two different MID conditions. Their respective effects appear depending on the acceptance number: 

\item Acceptance number $c=0$:  
The minimal binomial sampling plan $(42,0)$ is conservative in the sense that it remains formally admissible down to $N=43$. 
But for lot sizes smaller than $N_0=3063$, smaller samples are possible because it becomes easier to satisfy the only relevant condition \eqref{MID4} at LQ $\pb=7\%$. 
The minimum sample size as function of lot size is listed in Table \ref{tab_nmin_N_c0} and shown in Figs.~\ref{fig_n_N1N1_c0} and  \ref{fig_nmin_cALL_withISO}. 

\item Acceptance number $c=1$: The minimal binomial result $(66,1)$ is no longer globally conservative. 
Certainly, below $N_1=1947$, the necessary sample size first decreases for smaller lots because it becomes easier to satisfy the LQ condition \eqref{MID4}. 
But below a lot size of $N=139$, where the smallest possible sample size is $n=55$, now the AQL condition \MIDa{} takes over. 
It requires the sample size to grow quite sharply with decreasing lot size in order to ensure an acceptance probability at $\pa$ below $95\%$. 
This sharp upturn continues down to $N=100$, where the global lower bound \eqref{lowerboundN} is reached. 
The admissible (minimum) sample size as function of lot size is listed in Table \ref{tab_n_N1N2_c1} and plotted in Fig.~\ref{fig_n_N1N2_c1} (Fig.~\ref{fig_nmin_cALL_withISO}).

\item Acceptance number $c=2$: The minimal binomial result $(88,2)$ is not globally  conservative. Mainly the AQL condition \eqref{MID3} at \MIDa{} is relevant, requiring the minimum sample size to grow for decreasing lot size. This sharp upturn continues down to $N=200$, where the global lower bound \eqref{lowerboundN} is reached. 
The admissible (minimum) sample size as function of lot size is listed in Table \ref{tab_n_N1N2_c2}  and  plotted in Fig.~\ref{fig_n_N1N2_c2} (Fig.~\ref{fig_nmin_cALL_withISO}). 
Below roughly $N=500$, the consumer's risk drops far below the MID threshold such that $c=0,1$ acceptance sampling becomes more appropriate.

\item 
Figure \ref{fig_nmin_cALL_withISO} also shows  
the ISO 2859-1 sampling plans recommended officially for MID modules F and F1 \cite{Welmec8.10}. 
Clearly, these are not minimal for low acceptance numbers $c=0,1,2$. Moreover, their growth with sample size, roughly as $n \sim \sqrt{N}$, is not justified by the MID conditions as read in Sec.~\ref{sec:intro}. 
By contrast, the MID-optimised sample plans derived here for $c=0,1,2$ remain valid for arbitrarily large lots.   
\end{enumerate} 


\section{Simplified single-sampling scheme} 
\label{sec:summary}

\begin{table*}
\centering
\setlength{\tabcolsep}{0pt}
\renewcommand{\arraystretch}{1.2}
\rowcolors{4}{white}{gray!25}
\begin{tabular}{%
R{2.75em}R{1.5em}R{2.75em}p{0.5em}|
C{1.25cm}*{2}{C{3.25em}}|
C{1.25cm}*{2}{C{3.25em}}|
C{1.25cm}*{2}{C{3.25em}}}
&&&& 
\multicolumn{3}{c|}{Acceptance $c=0$} &
\multicolumn{3}{c|}{Acceptance $c=1$} &
\multicolumn{3}{c}{Acceptance  $c=2$} \\
\multicolumn{4}{c|}{Lot}  	& 
Sample & \multicolumn{2}{c|}{Risk [\%]} &
Sample & \multicolumn{2}{c|}{Risk [\%]} &
Sample & \multicolumn{2}{c}{Risk [\%]} \\
\multicolumn{4}{c|}{$N$} 	
& $n$ & $ \alpha \leq $  & $ \beta \geq $
& $n$ & $ \alpha \leq $  & $ \beta \geq $
& $n$ & $ \alpha \leq $  & $ \beta \geq $
\\
\hline 
    21  &to&  24 &&   20  & (43.4) &  0.69  
    &-&-&-&-&-&-\\ 
    25  &to&   31  &&   23  & (44.6) &   0.81
    &-&-&-&-&-&-\\ 
    32  &to&   41 &&   26 &   (40.6)  &   1.90 
     &-&-&-&-&-&-\\
    42  &to&   61  &&   30 &   (40.5) &   2.09 
     &-&-&-&-&-&-\\
    62  &to&  122  &&   35 &   (40.1)  &   2.30  
     &-&-&-&-&-&-\\
    123 &to&  248  &&   38 &    36.6  &  3.66  
                    &   - & - & - &-&-&-\\
    249 &to&  500   &&  40 & 35.4 & 4.21  
                    &  63 & 12.2 & 3.77   
                    &-&-&-\\
501     &to& 1000   &&   41 & 34.9  & 4.48  
                    &   65 & 13.4  & 4.23  
                    &   - &  - &  -  \\
1001    &to&  $\infty$ &&42 & 35.0   & 4.44  
               &   66 & 14.1  & 4.45 
               &   88 & 5.87  & 4.25 \\
\end{tabular} 
\caption{%
Proposal for a simplified single-sampling scheme optimised for MID modules F and F1  
based on the data of Tables \ref{tab_nmin_N_c0}, \ref{tab_n_N1N2_c1}, and \ref{tab_n_N1N2_c2}.
In all cases, the producer's risk is $\alpha>5\%$, and the consumer's risk is $\beta<5\%$, as required by MID; here only the variable upper (lower) bound for $\alpha$ ($\beta$) is shown. 
For isolated lots of size $N<100$, the producer's risk (listed in parentheses) has no operational significance because the AQL $\pa=1\%$ would correspond to less than a single defective item ($M<1$). 
} 
\label{MID_simp_scheme}
\end{table*}

\begin{figure*}\centering
\includegraphics[width=.9\textwidth]{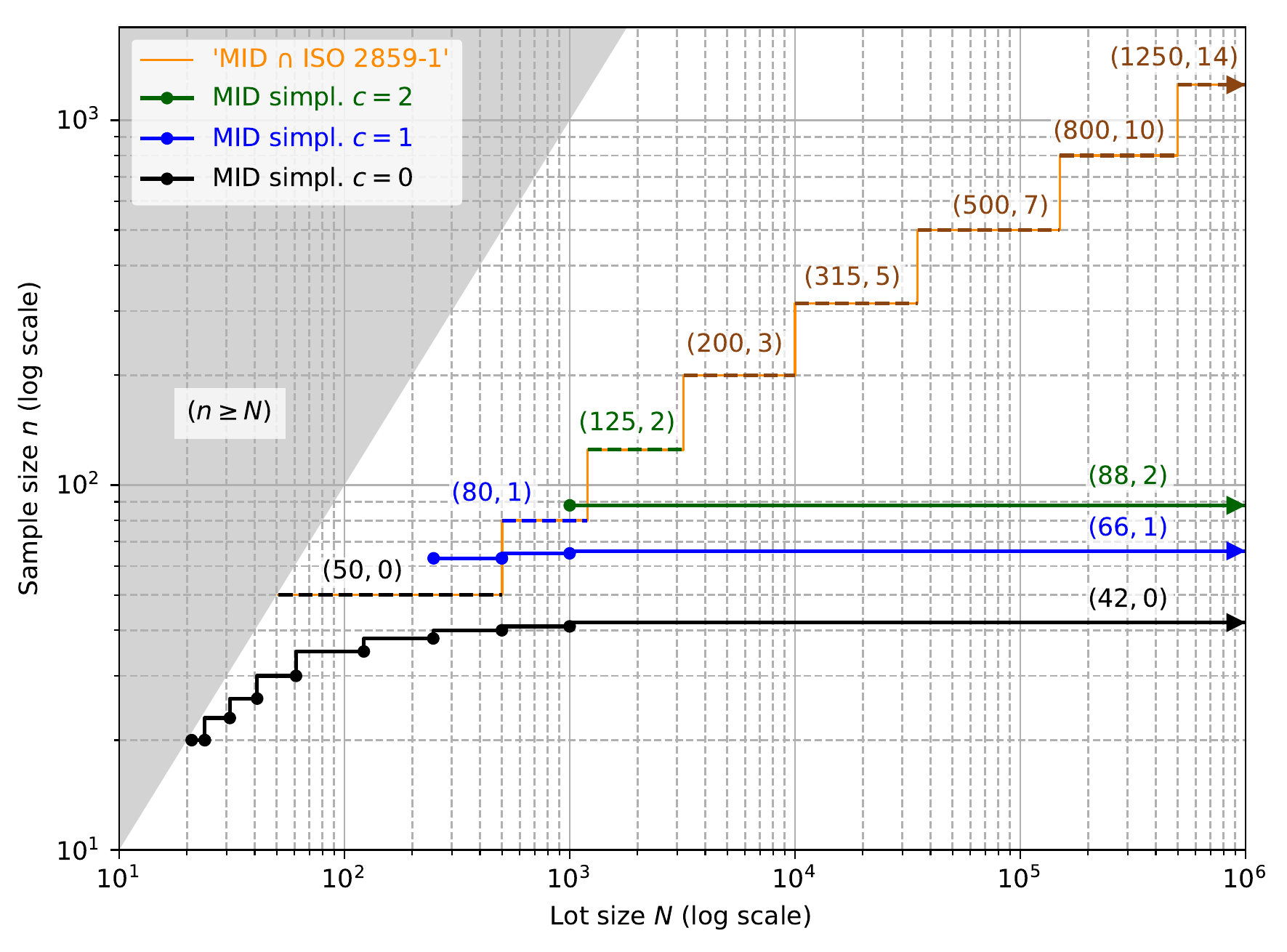}
\caption{%
Sample size as function of lot size for the simplefied, MID-optimised sampling scheme listed in Table \ref{MID_simp_scheme}, together with the single-sampling scheme from ISO 2859-1 \cite{Welmec8.10}, on a double-logarithmic scale. 
}   
\label{fig_nmin_cALL__simple_withISO}  
\end{figure*} 

The numerical minimisation of sample size  as function of lot size under MID conditions produces the sample plans $(n,c)$ that are listed in Tabs.~\ref{tab_nmin_N_c0}, \ref{tab_n_N1N2_c1} and ~\ref{tab_n_N1N2_c2} for $c=0,1,2$, respectively.  
Admittedly, these tables (and corresponding figures) are more complicated than the sample plans extracted from ISO 2859-1 that are recommended hitherto \cite{Welmec8.10}. 
It appears therefore advisable to condense the optimised plans into a simplified sampling system that is still (almost) optimal as far as the MID conditions are concerned, but efficient to use in practice. 
Table \ref{MID_simp_scheme} contains a proposal for such a sampling system, retaining the essential characteristics of the MID-minimised sampling plans. 
Figure \ref{fig_nmin_cALL__simple_withISO} represents this simplified scheme, to be compared to the exact data represented in Fig.~\ref{fig_nmin_cALL_withISO}.  
The main steps taken to arrive at the simplified scheme are:
\begin{itemize}
\item choosing a simple lower bound $N>1000$ common to the relevant, minimal binomial sampling plans $(42,0)$, $(66,1)$, and $(88,2)$; 
\item binning lot sizes into a manageable number of intervals common to all relevant acceptance numbers $c=0,1,2$.   
\end{itemize} 
Such a simplified proposal is slightly arbitrary in the sense that one may as well choose different lot-size intervals with, consequently, a different set of sample sizes. 
The present proposal aims at a reasonable compromise between the complexity of the scheme and its logistic efficiency.  
Independently of the finer details, the main feature of our proposal is that the sample size does not grow with the lot size above $N=1000$, while offering the full statistical protection required by the MID.

For small lot sizes $N\leq 248$, only zero-acceptance sampling is found to be practical, with advantageously small samples but rather elevated producer's risks. 
The larger the lot, the more options are offered for statistical sampling.
So one thing remains to be specified: Which acceptance number $c=0,1$  should be chosen above $N=248$, and which of $c=0,1,2$ above $N=1000$?  
Only if the producer knows in advance that the quality level of the lot is perfect ($p=0$ by production or strict exit controls), then $c=0$ with the smallest possible sample size is of course  preferable. 
If the quality level is finite ($p>0$) but unkown, then the choice of the sample plan is not uniquely determined by the MID conditions.  
Instead, the producer may decide which risk $\alpha$ is worth taking, depending on the lot's known or expected quality level, the production cost of each item, its inspection cost, etc.. 
If, on the one hand, production costs are low, but inspection costs are high, then an elevated producer's risk may be acceptable, and the smallest acceptance number $c=0$ is preferable with, consequently, the smallest sample size. 
If, on the other hand, production costs are high and inspection costs are low, then the producer's risk can be substantially reduced by raising the acceptance number to $c=2$, together with an altogether moderate increase of sample size.    
In order to arrive at a truly (or at least approximately) optimal choice, one would have to define an appropriate cost function and determine the optimal acceptance number by an in-depth cost-benefit analysis \cite{CBA}. 

\section{Outlook: alternative interpretation of the AQL  criterion}
\label{sec:outlook}

The AQL criterion of the MID, interpreted as in expressions \eqref{MID1},\eqref{MID3}, and \eqref{MIDa}, sets a \emph{lower} bound on the producer's risk. This is a somewhat curious condition, already from an economic and contractual point of view: why should the MID guarantee a one-sided protection of the consumer's interest at two points, without limiting the producer's risk at the acceptance quality level? 
Furthermore, this condition leads to a number of awkward mathematical properties that belie standard statistical knowledge in acceptance sampling. 
For example, the binomial approximation for the true hypergeometric distribution is known to be conservative in the sense that it provides larger samples than necessary for a certain finite lot size \cite{SchillingNeubauer2017}. 
This turns out not to be true here for $c\geq 1$ sampling in the lot-size range where the AQL condition dominates and requires larger samples than the binomial approximation. 

Indeed, standard textbooks formulate the AQL inequality usually the other way around, setting an upper bound also on the producer's risk---see, e.g., Eq.~(10.57) in \cite{Mathews2010}. 
Additionally, an upper bound on both the consumer's and producer's risk is compatible with the framework of hypothesis testing (see, e.g., section 2.2 in \cite{Klauenberg2017})\rev{, which could have important conceptual and operational benefits
\cite{Samohyl2018}. }
Details of MID-optimised sampling plans obtained under such an alternative interpretation of the MID's AQL criterion, however, are beyond the scope of the present work and will be presented in a forthcoming publication.

\section*{Acknowledgements}  
The author is indebted to Dr.~Katy Klauenberg for a critical reading of the manuscript and insightful comments.  

\appendix 

\section{Interpretation of MID criteria for finite lots}
\label{app:MID_finiteN}

\begin{figure*}\centering
\includegraphics[width=.9\textwidth]{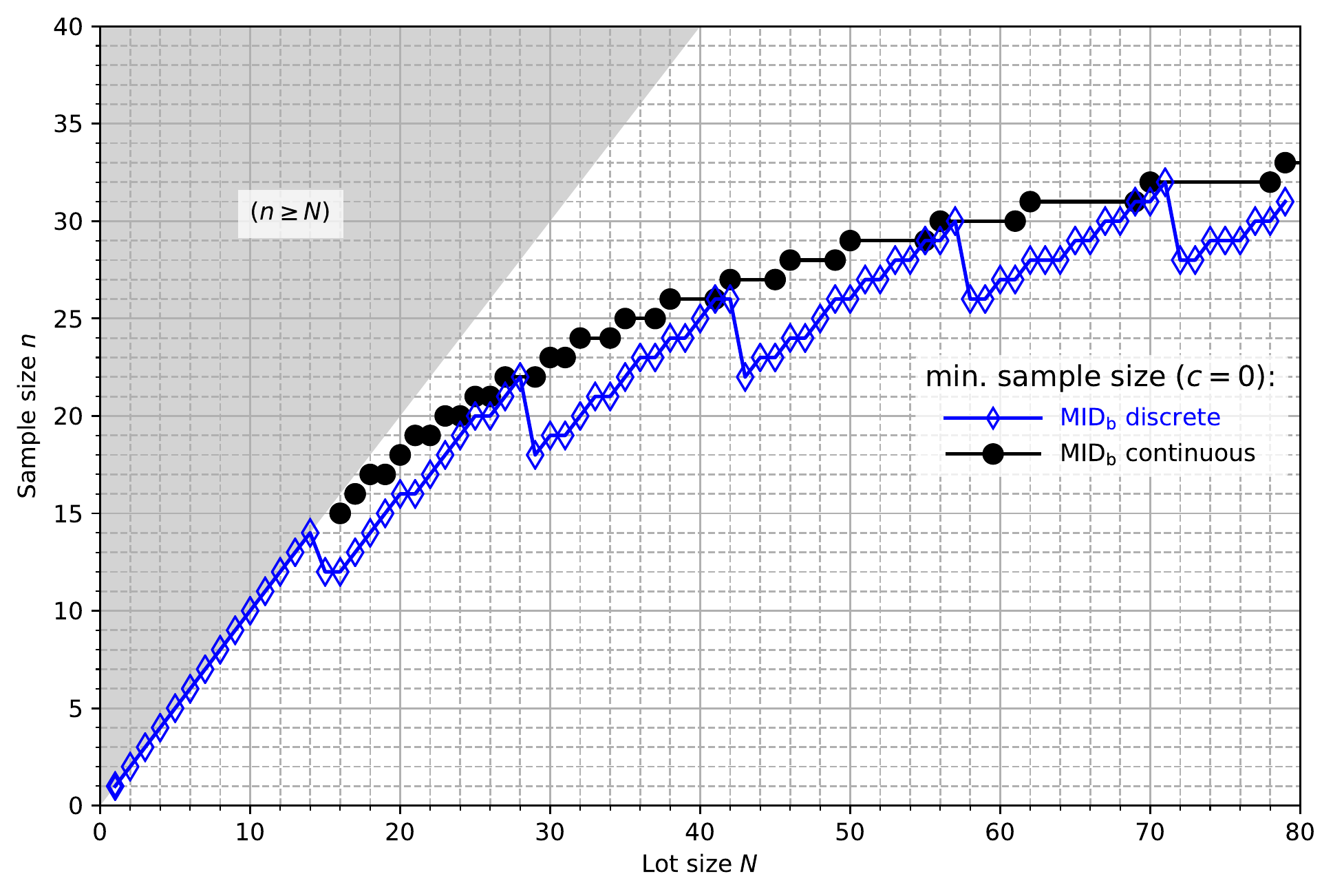}
\caption{%
Minimum sample size as function of lot size for zero-acceptance sampling $(c=0)$. The purely `discrete' MID criteria  \eqref{MID_finiteN_2} (blue diamonds) result in non-monotonic jumps, whereas the analytical coninuation introduced in Sec.~\ref{finiteN_c0.sec} (black circles) provides a non-decreasing, conservative bound for the minimum sample size.  
Since $c=0$, only the \MIDb{} criterion is actually relevant (cf.~Sec.~\ref{finiteN_c0.sec}). 
}   
\label{fig_n_N_allowed_c0}  
\end{figure*}

We need to discuss the applicability of the MID criteria, either \eqref{MID1} and \eqref{MID2} or 
\eqref{MID3} and \eqref{MID4}, for isolated lots of finite size, where quality levels and acceptance probabilities are discrete sets. 
Consider a lot of size $N$, containing an (unknown) number  $M\in\{0,1,\dots,N\}$ of non-conforming items. 
One can assign the quality level $p=M/N$ to this lot and discuss its acceptance probability $\Pac(p)$ under a certain sampling plan as function of the meaningful values 
\be 
p\in \mathcal{D}_N=\{0,\tfrac{1}{N},\tfrac{2}{N},\dots,\tfrac{N-1}{N},1\}.     
\label{allowed_p}
\ee
Now, the two special quality levels of the MID conditions, $\pa=\frac{1}{100}$ and $\pb=\frac{7}{100}$,  are only in the domain $\mathcal{D}_N$ 
if $N$ is an integer multiple of 100.
For all other lot sizes, $p_i\notin \mathcal{D}_N$ such that $\Pac(p_i N,N,n,c)$ as given by eq.~\eqref{hypergeom} is not defined, at least not in simple operational terms related to the sampling of a single lot, where $M$ must be an integer. 
As a consequence, the MID conditions \eqref{MID3} and \eqref{MID4} cannot be applied as such for deciding whether a sampling plan is admissible or not. 
And even if per chance $N$ is a multiple of 100, also the image of $\mathcal{D}_N$ under $\Pac$, the set $\Pac(\mathcal{D}_N)$ of the different acceptance probabilities, is now a discrete set that will generally not contain the two values $\Pa=\frac{95}{100}=\frac{19}{20} $ and $\Pb=\frac{5}{100}=\frac{1}{20}$. 
In other words, in most cases there will be no $q_i\in \mathcal{D}_N$ such that $\Pac(q_i)=P_i$, exactly. Thus, also the MID conditions \eqref{MID1} and \eqref{MID2} cannot be applied as such. 

It thus becomes clear that the wording of the MID implicitly assumes a continuous description, which applies only to process sampling with a formally infinite lot size (called ``type B'' in \cite{SchillingNeubauer2017,Mathews2010}). 
However, the MID product testing in modules F and F1 typically involves isolated lots of finite size, where the continuity of type B testing cannot be taken for granted. 
In a first attempt to render the MID conditions meaningful for finite lot sizes, we have tried to interpret the wording ``corresponding to a probability of acceptance of $P_i$'' as meaning ``corresponding to a probability of acceptance of \emph{at least} $P_i$''.%
\footnote{In the continuous case, this follows already from the monotonicity and continuity of $\Pac$: There is exactly one $q_i$ such that $\Pac(q_i)=P_i$ and for which MID implies $q_i<p_i$. Now take any $p$ such that $\Pac(p) > \Pac(q_i)$. The monotonicity ($p>q_i\Rightarrow \Pac(p)\leq\Pac(q_i)$) in its negated form $\Pac(p) > \Pac(q_i) \Rightarrow p\leq q_i$ then implies together with $q_i<p_i$ the r.h.s. $p<p_i$. So the MID conditions could have been formulated with a ``probability of \emph{at least} $P_i$'' from the start.} 
Then, the two MID conditions \eqref{MID1} and \eqref{MID2} are rather 
\begin{align}
 \Pac(p) &\geq P_i \Rightarrow p < p_i, \quad i=\text{a,b},   \label{MID_finiteN}
\end{align}  
which can be tested for all $p\in \mathcal{D}_N$. A logically equivalent, but more practical criterion that can be readily evaluated with a computer is obtained by the negation of \eqref{MID_finiteN}, namely 
\begin{align}
p & \geq p_i \Rightarrow \Pac(p) < P_i, \quad i=\text{a,b}.   \label{MID_finiteN_2}
\end{align}  
Under these premises, in order to test the admissibility of a certain sampling plan, one only needs to take the first allowed quality levels $p\in\mathcal{D}_N$ larger than or equal to $\pa$ and $\pb$, respectively,  and check whether $\Pac(p)$ at these points is smaller than $\Pa$ and $\Pb$, respectively. 
If that is the case, then the plan is approved (monotonicity guarantees that even larger $p$ cannot yield larger values of $\Pac$), and if not, it is rejected.

While this interpretation is logically satisfying and economical to evaluate, it leads to inconsistencies due to discretisation effects. 
Indeed, when the lot size is increased such that one quality level $p_j\in\mathcal{D}_N$ drops below $\pb$, the next higher quality level $p_{j+1}$ becomes  relevant such that suddenly a smaller sample may become allowed. 
Figure \ref{fig_n_N_allowed_c0} shows the $c=0$ minimum sample size as function of lot size  resulting from the purely  `discrete' criterion \eqref{MID_finiteN_2}. 
It features  a prominent saw-tooth structure where raising the lot size at certain points (e.g., from $N=42$ to 43) suddenly lowers the minimum sample size (e.g., from $n=26$ to 22).

Such an erratic dependence of sample size on lot size is arguably not a desirable feature of an acceptance sampling plan. Therefore, we have proposed in Sec.~\ref{finiteN_c0.sec} to extend the acceptance probabilities in the standard manner to non-integer item numbers, formally adopting a type-B testing scenario that allows us to evaluate $\Pac(p)$ at $\pa$ and $\pb$ for any lot size and thus to use the MID criteria \eqref{MID3} and \eqref{MID4} just as in a truly continuous case. 
The advantage is two-fold: 
First, the resulting sample size is conservative, namely, never smaller than prescribed by the discrete criterion \eqref{MID_finiteN_2} at the consumer's LQ point, where also the binomial model for infinite lots is conservative.  
Second, the minimum lot size resulting from the consumer's LQ point condition is a non-decreasing function of lot size, as evident from Fig.~\ref{fig_n_N_allowed_c0} (see also Fig.~\ref{fig_n_N1N1_c0} and the relevant portion of Fig.~\ref{fig_n_N1N2_c1}).

\end{multicols}

\end{document}